\documentclass [11pt, pdftex]{amsart}

\usepackage[foot]{amsaddr}
\usepackage {amsmath, amsthm, amssymb, amscd, mathrsfs, graphicx, color, url, enumerate, epsf}
\usepackage {amsmath, amssymb, amscd, mathrsfs, graphicx, color, url, enumerate, epsf}
\usepackage[all,knot,arc]{xy}
\usepackage{cite}

\usepackage{fullpage}

\newcommand{\remove}[1]{}

\newtheorem {theorem}{Theorem}[section]
\newtheorem {lemma}[theorem]{Lemma}
\newtheorem {proposition}[theorem]{Proposition}
\newtheorem {corollary}[theorem]{Corollary}

\newtheorem {definition}[theorem]{Definition}
\newtheorem {remark}[theorem]{Remark}
\newtheorem {example}[theorem]{Example}

\def\minimal{\mathit{minimal}}

\newcommand{\ignore}[1]{}

\renewcommand\th{^{\text{th}}}
\newcommand\st{^{\text{st}}}

\newcommand\R{\mathbb{R}}

\newcommand\T{\mathscr{T}}

\def\fast {\mathit{fast}}
\def\part {\mathit{part}}

\def\view{\mathit{view}}

\def\Runs{\mathcal{R}}
\def\IIS {\operatorname{IIS}}
\def\SM {\operatorname{SM}}

\def\A {\mathbb{A}}

\def\St {\operatorname{St}}
\def\st {\operatorname{st}}
\def\s {\mathbf{s}}
\def\t {\mathbf{t}}

\def\Chr{\operatorname{Chr}}

\def\WF{\mathit{WF}}

\def\Res{\mathit{Res}}
\def\OF{\mathit{OF}}
\def\O {\mathcal{O}}
\def\I {\mathcal{I}}

\def\Bary{\operatorname{Bary}}
\def\ipart{\infty\text{-}\mathit{part}}

\def\slow{\mathit{slow}}

\def\ord{\operatorname{ord}}
\def\V{\mathcal{V}}

\def\Skel {\operatorname{Skel}}
\def\x{\mathbf{x}}

\def\adv{\operatorname{adv}}

\newfont{\mycrnotice}{ptmr8t at 7pt}
\newfont{\myconfname}{ptmri8t at 7pt}

\begin{document}

\bibliographystyle{abbrv}

\title{A Generalized Asynchronous Computability Theorem}

\author[Gafni]{Eli Gafni$^1$}
\address{$^1$Computer Science Department, UCLA, 3731F Boelter
  Hall\\ Los Angeles, CA 90095, USA (Corresponding author)}
\email{eli@ucla.edu}

\author[Kuznetsov]{Petr Kuznetsov$^2$}
\address{$^2$T\'el\'ecom ParisTech, INFRES, 46 Rue Barrault, 75013 Paris, France}
\email{petr.kuznetsov@telecom-paristech.fr}

\author[Manolescu]{Ciprian Manolescu$^3$}
\thanks {CM was supported by NSF grant DMS-1104406.}
\address {$^3$Department of Mathematics, UCLA, 520 Portola Plaza\\ Los Angeles, CA 90095, USA}
\email {cm@math.ucla.edu}

\maketitle

\begin{abstract}
We consider the models of distributed computation defined as subsets of the
runs of the iterated immediate snapshot model. 
Given a task $T$ and a model $M$, 
we provide topological conditions for $T$
to be solvable in $M$.

When applied to the wait-free model, our conditions result in the celebrated Asynchronous
Computability Theorem (ACT)
of Herlihy and Shavit.

To demonstrate the utility of our characterization, we consider a task that
has been shown earlier to admit only a very complex $t$-resilient solution.
In contrast, our generalized computability theorem
confirms its $t$-resilient solvability in a straightforward manner. 
\end{abstract}

\section{Introduction}

This paper characterizes task solvablility in models of distributed
computing, where processes communicate via reading from and writing to
a shared memory. We treat a model as a set of \emph{runs}, i.e., interleaving of read and
write steps issued by different processes.

What do we mean by a characterization? 
We say that a task $T$ is solvable in a model $M$, if there exists a
\emph{protocol} by which, in every run of $M$, 
each process taking sufficiently many steps eventually \emph{outputs},
so that the outputs satisfy the task's specification with respect to
the provided inputs. 
The conventional definition of solvability is therefore \emph{operational},
based on the existence of a protocol.
A topological characterization replaces the operational definition with 
the existence of a continuous map between topological spaces, capturing
the sets of possible inputs  and outputs of the task.
This topological characterization may provide insights about 
the (in)solvability of the task that are not easy to grasp operationally~\cite{HS93,BG93b,SZ00}.

In 1993, Herlihy and Shavit~\cite{HS93,HS99} characterized read-write
communication with no restrictions on the runs; this is referred to as the
\emph{wait-free} model. They formulated the Asynchronous Computability
Theorem (ACT) stating that a task $T$ is wait-free solvable
\emph{if and only if} there exists a simplicial, chromatic map from a subdivision
of simplexes of an appropriately defined input simplicial complex to an
appropriately defined output simplicial complex, satisfying the
specification of $T$. 

The original proof of ACT is given in~\cite{HS99} directly for the conventional
read-write shared-memory model (referred to as the \emph{standard
  shared-memory} and denoted SM), 
where memory consists of persistent
objects which can be written and read by a given process arbitrarily
often.
However, the original proof can be simplified by casting the
problem to the \emph{iterated immediate-snapshots}
(IIS)~\cite{BG97} shared memory model, in which processes march through a sequence of
Immediate-snapshots (IS) tasks invoking the next one with their output from
the previous one. 

The IIS model can be thought of as convenient mathematical tool to analyze and
understand distributed computing, contrasting with the more realistic but
less convenient standard shared-memory model (SM).
Casting distributed computation in SM to IIS is not unlike 
analyzing electromagnetic communication in the complex-number domain. 
We start with the real-world, we move the reasoning to an abstract mathematical world,
and then we translate the results back to reality.

The proof of ACT in~\cite{HS99} can thus be simplified along the
following lines:
\begin{enumerate}
\item[(1)] 

The wait-free runs in standard shared-memory can be simulated in the IIS
model~\cite{BG97,GR10-opodis}, and the IIS model can be simulated by the standard shared-memory model \cite{BG93b}.
Thus all we need is to characterize task solvability in IIS.

\item[(2)] If a task $T$ is wait-free solvable in the IIS model, then there
  exists an integer $k_T$ such that after the first $k_T$ invocations of
  immediate snapshots, each process can output in $T$. 
  This can be shown by a compactness argument or K\"onig's lemma (cf.,
  e.g., \cite{BG97}).

\item[(3)] Solving a task by $k_T$ immediate snapshots can be
  interpreted topologically as constructing a simplicial map from the $k_T$-th
  standard chromatic subdivision of the input complex to the output
  complex. This is because each immediate snapshot can be represented
  by a standard subdivision~\cite{Lin10, Koz12}.

\item[(4)] A chromatic map from an arbitrary subdivision of a complex can
  be approximated by a chromatic map from an iterated standard chromatic
  subdivision. 
This is a chromatic version
  of the standard simplicial approximation theorem, and a proof can be found in
  Section 5 of~\cite{HS99}.  It also follows from an operational argument in~\cite{BG97}.

\end{enumerate}

In this paper we extend this method of proof 
to models representing proper subsets of the runs of the wait-free
standard shared-memory model.

To this end, we start with a formal definition of the very notion of
solving a task in IIS. Surprisingly, no such definition appeared in
the literature until now.\footnote{Some specific \emph{compact} subsets of IIS runs were
  formally treated in~\cite{RRT08}.} In particular, we introduce the notions of 
\emph{participating} (taking at least one step) and {\em infinitely
  participating} (taking infinitely many steps) processes in an IIS run.

Further, to benefit from Step~(1) in the list above, we need a mapping between SM and IIS, that preserves the notions of participating and infinitely participating sets. Furthermore, in addition to this  forward simulation $F:~\SM\rightarrow \IIS$ we need a backward
simulation $B:~\IIS\rightarrow \SM$, such that 
for all $M\subseteq \SM$ we have $B(F(M)) \subseteq M$. We also ask for the restriction of $B$ to the image of $F$ to preserve the notions of participating and infinitely participating sets.
Given that the standard simulations~\cite{BG93a,BG97} do not meet
these requirements, we employ a new two-way simulation presented
in~\cite{BGK14}. This applies to a large set of \emph{adversarial}
SM models~\cite{DFGT11,Kuz12},
specified by sets of processes that can be infinitely participating in
a model run.
 
Thus, we can cast an adversarial SM model to its equivalent
model in IIS.  Moreover, the IIS model is, in a strict sense, richer
than SM: multiple IIS runs collapse into a single SM run by the simulation \cite{BGK14}.  
Our characterization of IIS task solvability applies to any {\em sub-IIS model} (that is, a subset of the runs in the IIS model), including those that have no equivalents in SM.

Imitating Step~(2) in sub-IIS models is in general
impossible. For example, in the $1$-resilient 
$3$-process model, the task of $2$-set agreement can be easily solved.  
However, every such solution has a run in which two processes fail
and the remaining process never outputs~\cite{HS93,BG93b,SZ00}.
In every $1$-resilient extension of a finite prefix of this run, 
all infinitely participating processes output.
Since such a prefix may have an arbitrary length, 
a uniform bound $k_T$ on the number of steps sufficient 
to output in a run of the $1$-resilient sub-IIS model does not exist. 
This observation is related to the fact that the model is {\em
  non-compact}, with respect to a metric that will be defined in Section~\ref{sec:top}. 

Therefore, our main theorem, which we call the \emph{generalized} asynchronous
computability theorem (GACT), 
is a generalization of Step (3) in the outline above, without  
relying on Step~(2). Instead, the characterization  proposes to ``approximate'' a
non-compact model by a sequence of compact models. The sequence converges to a superset of the 
target model. Thus, if the task is solvable by each compact model in the sequence, it is solvable by the target model.
Each compact model is represented as subcomplex of a subdivided
simplex. 
Hence GACT deals with a sub-complex of a subdivided simplex instead of a
subdivided simplex (as ACT does), and instead of saying that there exists 
a single subdivision (as in ACT)  
it requires the existence of a sequence of sub-complexes.

Step~(4) now applies individually to each subdivision in the
sequence. In this paper we do not deal with this step, because the formulation would be too cumbersome, and because it is not necessary for our examples. Nevertheless, we believe that our theorem can be extended to arbitrary rather than standard chromatic subdivisions.

ACT turned out to be 
an essential tool in distributed
computing~\cite{GK99-undecidable,HR97,HR13,HR10-shell,HR00-spans,CHR12}. 
We show that our GACT holds that promise too.
We consider a task $T$, solvable $t$-resiliently, but (to our knowledge) only with a very
involved algorithm~\cite{G98}. 
In contrast, by applying the methods developed in this paper, we show that determining the $t$-resilient
solvability of $T$ is relatively simple.

The paper is organized as follows: 
In Section~\ref{sec:models} we describe the IIS model and give some examples of sub-IIS models. In Section~\ref{sec:topdef} we review some notions from combinatorial topology. In Section~\ref{sec:tasks} we review the topological definition of a task, and explain what it means for a task to be solvable in a model. In Section~\ref{sec:top} we describe sub-IIS models topologically. In Section~\ref{sec:gact} we prove our main result, GACT. 
In Section~\ref{sec:wf} we explain how GACT gives back the well-known
ACT in the wait-free case. In Section~\ref{sec:proofs} we introduce a new topological tool: a version of the simplicial approximation for infinite chromatic complexes. Using this tool, in Section~\ref{sec:lc} we show how GACT can be applied to a class of tasks called link connected; in particular, we use GACT to prove that a particular task can be solved in the
$t$-resilient model. 
In Section~\ref{sec:related} we recall some related work, and
in Section~\ref{sec:conc} we draw the conclusions.

\section{Sub-IIS models}
\label{sec:models}

In this section, we describe our perspective on the \emph{Iterated
  Immediate Snapshot} (IIS) model~\cite{BG97} and give examples of
sub-IIS models. 

\subsection{The IIS model} \label{sec:IIS}
Suppose we have $n+1$ processes $p_0, p_1, \dots, p_n$. 
A run $r$ in IIS is a sequence of \emph{non-empty} sets
of processes $S_1 \supseteq S_2 \supseteq \dots$, with each $S_k
\subseteq \{p_0, \dots, p_n\}$ consisting of those processes that
participate in the $k$th iteration of immediate snapshot
(IS). Furthermore, each $S_k$ is equipped with an ordered partition:
$S_k = S_k^1 \cup \dots \cup S_k^{n_k}$ (for some $n_k \leq n$),
corresponding to the order in which processes are invoked in the
respective IS. 

Let $\Runs$ be the set of runs in IIS. Fix a run $r \in \Runs$, with $r = S_1,S_2, \ldots$ as above. The processes $p_i \in S_1$ are called {\em participating}.  
If $p_j$ appears in all the sets $S_k$, we say that $p_j$ is
{\em infinitely participating} in $r$. 
The sets of participating and infinitely participating processes in a run $r$ are denoted $\part(r)$ and $\ipart(r)$, respectively.

If either $k=0$ or $p_i \in S_k$ for some $k \geq 1$, then we define a
set called the \emph{$k$th view of ${p_i}$ in the run $r$},
$\view(p_i, k)$, recursively, as follows:
\begin{enumerate}
\item $\view({p_i}, 0)=\{p_i\}$;
\item For $k \geq 1$, the view of ${p_i} \in S_k^j \subseteq S_k$ is $\view({p_i}, k)=\{\view({p_s}, {k-1}) \mid p_s \in S_k^1 \cup \dots \cup S_k^j \}$.
\end{enumerate}

Our definitions can be interpreted operationally as follows.
Every process proceeds through an infinite series of one-shot
immediate snapshot (IS)
instances~\cite{BG93a}: $IS_1,IS_2\ldots$.  
Then $S_k$ is interpreted as the set of processes
accessing memory $IS_k$, and each $S_k^j$ is the set of processes
obtaining the same \emph{view} after accessing $IS_k$.
Recall that in IS, the view of a process $p_i\in S_k^j$ is defined by the values
written by the processes in $S_k^1 \cup \dots \cup S_k^j$.     

The original definition of IIS~\cite{BG97} can be thought of as the
variant of our model, in which we impose the condition $S_1=S_2 =\dots
= \{p_0, \dots, p_n\}$, i.e., every process is infinitely
participating.  What is the advantage of our new, more refined, definition of IIS? 
It allows a run to be extended to more processes without changing the
views of the already existing processes. 
For instance, in the run $r=\{\{p_0\}\},\{\{p_0\}\},\ldots$ we have  $\part(r)=\ipart(r)=\{p_0\}$.
In the run $r'=\{\{p_0\}, \{p_1\}\}, \{\{p_0\}, \{p_1\}\}, \{\{p_0\},
\{p_1\} \},\ldots$ we have $\part(r')=\ipart(r')= \{p_0,p_1\}$. 
However, $p_0$ cannot tell whether it is in $r$ or in $r'$, because
the corresponding views of $p_0$ are the same in both runs. In this situation, we say that $r'$ is an extension of $r$.

Formally, we say that a run $r' = S_1', S_2', \dots$ is an {\em
  extension} of a run $r = S_1, S_2, \dots$, and we write $r\leq r'$,
if (i) $S_j\subseteq S_j'$ for all $j$, and (ii) the views of the
processes in $\part(r)$ are the same in $r'$ as in $r$. 
This defines a partial order on $\Runs$.

If $r$ is a run, let $\minimal(r)$ be the smallest run $r_0$ such
that $r_0 \leq r$ (that is, for all $r' \leq r$, we have $r_0 \leq
r'$). It is not difficult to see that $\minimal(r)$ exists and is unique. We then define $\fast(r)=
\ipart(\minimal(r))$. We define $\slow(r)$ to be the complement set of
$\fast(r)$. 

Intuitively, $\fast(r)$ is the largest set of processes that ``see''
each other (appear in each other's view) infinitely often in $r$. 
In other words, for all $p_i,p_j\in \fast(r)$ and all $k\geq 0$, there
exists $\ell\geq k$ such that $\view(p_i,k)$ appears in
$\view(p_j,\ell)$. 

\subsection{Examples of models} We define a {\em sub-IIS model} $M$ to be any subset of $\Runs$.

\begin{example}
\label{ex:wf}
The {\em wait-free} (or {\em completely asynchronous}) model $\WF$ is
the set $\Runs$ itself. The interpretation of $\WF$ is that anything
can happen (all sorts of step interleavings are allowed). 
\end {example}

\begin{example}
\label{ex:res}
For $t\leq n$, the {\em $t$-resilient model} $\Res_t$ consists of the runs $r \in \Runs$ such that $|\fast(r)| \geq n+1-t.$ This is the model in which at most $t$ processes are slow. 
\end{example}

\begin{example}
\label{ex:obstruction-free}
For $k\leq n+1$, the {\em $k$-obstruction-free model} $\OF_k$ consists
of all the runs $r$ in which no more than $k$ processes are fast,
i.e., $|\fast(r)| \leq k.$  This model was
previously discussed in \cite{Gaf08-concurrency},  following a suggestion of Guerraoui. \end{example}

\begin{example}
\label{ex:adv}
More generally, consider the {\em model with adversary $\A$} \cite{DFGT11}, which we denote by $M^{\adv}(\A)$. Here, $\A$ is any subset of the power set of $\{0, 1, \dots, n\}$.  We then define $M^{\adv}(\A)$ to consist of all runs $r$ such that $ \slow(r) \in \A$. 
\end{example}

\section{Topological definitions}
\label{sec:topdef}

Before moving forward, we need to review several notions from topology. We will assume that the reader has a basic knowledge of metric spaces (open sets, continuity, compactness), as in \cite[Chapter 7]{Royden}.

\subsection{Simplicial complexes}
A good reference for the material in this section is Chapter 3 in \cite{Spanier}.

A {\em simplicial complex} is a set $V$, together with a collection $C$ of finite nonempty subsets of $V$ 
such that:
\begin{enumerate}[(a)]
\item For any $v \in V$, the one-element set $\{v\}$ is in $C$;
\item If $\sigma \in C$ and $\sigma' \subseteq \sigma$, then $\sigma' \in C$.
\end{enumerate}

The elements of $V$ are called {\em vertices}, and the elements of $C$ are called a {\em simplices}. We usually drop $V$ from the notation, and refer to the simplicial complex as $C$. 

A simplicial complex $C$ is called {\em finite} if the collection $C$ is finite. A weaker notion is {\em locally finite}: $C$ is said to be locally finite if every vertex of $C$ belongs to only finitely many simplices  in $C$. For simplicity, we will assume that our complexes are locally finite.

A subset of a simplex is called a {\em face} of that simplex.

A {\em subcomplex} of $C$ is a subset of $C$ that is also a simplicial complex.

The {\em dimension} of a simplex $\sigma \in C$ is its cardinality minus one.  The $k$-skeleton of a complex $C$, denoted $\Skel^k C$, is the subcomplex formed of all simplices of $C$ of dimension $k$ or less.

A simplicial complex $C$ is called {\em pure} of dimension $n$ if $C$ has no simplices of dimension $> n$, and every $k$-dimensional simplex of $C$ (for $k < n$) is a face of an $n$-dimensional simplex of $C$.

Given a simplex $\sigma \in C$, we denote by $\st \sigma$ the {\em open star} of $\sigma$, that is, the set of all simplices in $C$ that have $\sigma$ as a face. The {\em closed star} of $\sigma$, denoted $\St \sigma$, is the smallest simplicial complex that contains $\st \sigma$. The difference $(\St \sigma) \setminus (\st \sigma)$ is called the {\em link} of $\sigma$.

Let $A$ and $B$ be simplicial complexes. A map $f: A \to B$ is called {\em simplicial} if it is induced by a map on vertices; that is, $f$ maps vertices to vertices, and for any $\sigma \in A$, we have
$$ f(\sigma) = \bigcup_{v \in \sigma} f(\{v\}).$$
A simplicial map $f$ is called {\em noncollapsing} (or {\em dimension-preserving}) if $\dim f(\sigma) = \dim \sigma$ for all $\sigma \in A$. 

Any simplicial complex $C$ has an associated {\em geometric realization} $|C|$, defined as follows. Let $V$ be the set of vertices in $C$. As a set, we let $C$ be the subset of $[0,1]^V = \{ \alpha : V \to [0,1]\}$ consisting of all functions $\alpha$ such that $\{ v \in V \mid \alpha(v) > 0 \} \in C$ and $\sum_{v \in V} \alpha(v) = 1$. For each $\sigma \in C$, we set 
$|\sigma| = \{ \alpha \in |C| \mid \alpha(v) \neq 0 \Rightarrow v \in \sigma \}.$
Each $|\sigma|$ is in one-to-one correspondence to a subset of $\R^n$ of the form $\{(x_1, \dots, x_n) \in [0,1]^n \mid \sum x_i = 1\}.$ We put a metric on $|C|$ by $d(\alpha, \beta) = \sum_{v \in V} |\alpha(v) - \beta(v)|.$ 

Given a simplicial map $f: A \to B$, there is an associated continuous, piecewise linear map $|f|: |A| \to |B|$, defined by the formula
$$ |f|(\alpha)(v') = \sum_{f(v) = v'} \alpha(v).$$

A nonempty complex $C$ is called {\em $k$-connected} if, for each $m \leq k$, any continuous map of the $m$-sphere into $|C|$ can be extended to a continuous map over the $(m+1)$-disk.

A {\em subdivision} of a simplicial complex $C$ is a simplicial complex $C'$ such that:
\begin{enumerate}
\item The vertices of $C'$ are points of $|C|$.

\item For any $\sigma' \in C'$, there exists $\sigma \in C$ such that $\sigma' \subset |\sigma|$.

\item The piecewise linear map $|C'| \to |C|$ mapping each vertex of $C'$ to the corresponding point of $C$ is a homeomorphism.
\end{enumerate}

In particular, every complex $C$ admits a {\em barycentric subdivision} $\Bary(C)$, defined as follows. The vertices of $\Bary(K)$ are the barycenters of the simplices of $C$ (in the geometric realization). The simplices of $\Bary(K)$ correspond to ordered sequences $(\sigma_0, \dots, \sigma_m)$ of simplices of $C$, where $\sigma_i$ is a face of $\sigma_{i+1}$; the barycenters of $\sigma_i$ are then the vertices of the corresponding simplex in $\Bary(K)$.

By iterating this construction $k$ times we obtain the $k$th barycentric subdivision, $\Bary^k(C)$.

\subsection{Chromatic complexes}
We now turn to the chromatic complexes used in distributed computing, and recall some notions from \cite{HS99}.

Fix $n \geq 0$. The {\em standard $n$-simplex} $\s$ has $n+1$ vertices, in one-to-one correspondence with $n+1$ {\em colors} $0, 1, \dots, n$. A face $\t$ of $\s$ is specified by a collection of vertices from $\{0,  \dots,  n\}$. We view $\s$ as a complex, with its simplices being all possible faces $\t$. Note that the open star of a face $\t$ is $\st \t = \{\t' \mid \t \subseteq \t' \},$ while the closed star of any face is the whole simplex $\s$.

A {\em chromatic complex} is a simplicial complex $C$ together with a noncollapsing simplicial map $\chi: C \to \s$. Note that $C$ can have dimension at most $n$. We usually drop $\chi$ from the notation. We write $\chi(C)$ for the union of $\chi(v)$ over all vertices $v \in C$. Note that if $C' \subseteq C$ is a subcomplex of a chromatic complex, it inherits a chromatic structure by restriction. 

In particular, the standard $n$-simplex $\s$ is a chromatic complex, with $\chi$ being the identity.

Every chromatic complex $C$ has a {\em standard chromatic subdivision} $\Chr C$. Let us first define $\Chr \s$ for the standard simplex $\s$. The vertices of $\Chr \s$ are pairs $(i, \t)$, where $i \in \{0,1 ,\dots, n\}$ and $\t$ is a face of $\s$ containing $i$. We let $\chi(i, \t) = i$. Further, $\Chr s$ is characterized by its $n$-simplices; these are the $(n+1)$-tuples $((0,\t_0), \dots, (n, \t_n))$ such that:
\begin{enumerate}[(a)]
\item For all $\t_i$ and $\t_j$, one is a face of the other;
\item If $j \in \t_i$, then $\t_j \subseteq \t_i$. 
\end{enumerate} 
The geometric realization of $\s$ can be taken to be the set $\{\x=(x_0, \dots, x_n) \in [0,1]^{n+1} \mid \sum x_i = 1\},$ with the vertex $i$ corresponding to the point $\x^i$ with $i$ coordinate $1$ and all other coordinates $0$. Then, we can identify a vertex $(i, \t)$ of $\Chr \s$ with the point
$$\frac{1}{2k-1} \x_i + \frac{2}{2k-1}  \Bigl( \sum_{\{j \in \t \mid j \neq i\}} \x_j \Bigr) \  \in |\s| \subset \R^{n+1},$$
where $k$ is the cardinality of $\t$. (Compare \cite[Definition 5.7]{HS99}.) Thus, $\Chr \s$ becomes a subdivision of $\s$ and the geometric realizations are identical: $|\s|=|\Chr \s|$. 

Next, given a chromatic complex $C$, we let $\Chr C$ be the subdivision of $C$ obtained by replacing each simplex in $C$ with its chromatic subdivision. Thus, the vertices of $\Chr C$ are pairs $(p, \sigma)$, where $p$ is a vertex of $C$ and $\sigma$ is a simplex of $C$ containing $p$. If we iterate process this $m$ times we obtain the $m\th$ chromatic subdivision, $\Chr^m C$. 

Let $A$ and $B$ be chromatic complexes.  A simplicial map $f: A \to B$  is called a {\em chromatic map} if for all vertices $v \in A$, we have $\chi(v) = \chi(f(v))$. Note that chromatic map is automatically noncollapsing. A chromatic map has chromatic subdivisions $\Chr^m f: \Chr^m A \to \Chr^m B$. Under the identifications of topological spaces $|A| \cong |\Chr^m A|, |B| \cong |\Chr^m B|,$ the continuous maps $|f|$ and $|\Chr^m f|$ are identical.

A {\em chromatic multi-map} between $A$ and $B$ is a map $\Delta: A \to 2^B$ that, for any $m \leq n$, takes every $m$-simplex of $A$ to a pure $m$-dimensional subcomplex of $B$, such that: (i) For every simplex $\sigma$ of $A$, we have $\chi(\sigma) = \chi(\Delta(\sigma))$, and (ii) For all simplices $\sigma, \tau \in A$,  we have
$ \Delta(\sigma \cap \tau) \subseteq \Delta(\sigma) \cap \Delta(\tau).$ In particular, if $\sigma'$ is a face of $\sigma$, then $\Delta(\sigma') \subseteq \Delta(\sigma).$

\section{Tasks}
\label{sec:tasks}

\subsection{Definitions}
\label{sec:deftasks}

A {\em task} $T = (\I, \O, \Delta)$ on $n+1$ processes $\{p_0, \ldots,
p_n\}$ consist of two finite, pure $n$-dimensional chromatic complexes $\I$ and $\O$, together
with a chromatic multi-map $\Delta: \I \to 2^{\O}$. The {\em input complex} $\I$ specifies the possible input values, the {\em output complex} $\O$ specifies the possible output values, and $\Delta$ describes which output values are allowed for a given input. The colors specify to which process each input or output value corresponds.

A task is called {\em input-less} if the input complex is the standard
simplex $\s$, colored by the identity. 
Then each process starts with input only its own id.\footnote{Note that in the definition of a multi-map we allowed images to be empty. This is somewhat non-standard, as it means that processes in a task do not have to output. If one prefers to avoid that, for every task $T=(\I, \O, \Delta)$ we can construct a new, equivalent task $T^+ = (\I^+, \O^+, \Delta^+)$ as follows. We let $\I^+=\I$. The output complex $\O^+$ is obtained from $\O$ by adding extra vertices $v_0, \dots,v_n$ (with $v_i$ corresponding to ``no output'' for the process $i$); moreover, for each simplex $\sigma$ in $\O$, we add an $n$-simplex $\sigma^+$ in $\O^+$ by adjoining vertices $v_i$ for the colors $i$ not represented in $\sigma$. Finally, we let $\Delta^+(\tau) = (\Delta(\tau))^+$.}

\subsection{Affine tasks}
Many examples of input-less tasks can be constructed as follows. Let $L \subseteq \Chr^k \s$ be a pure $n$-dimensional subcomplex of the $k\th$ chromatic subdivision of $\s$, for some $k$. For each face $\t \subseteq \s$, the intersection $L \cap \Chr^k \t$ is a subcomplex of $\Chr^k \s$; we assume that this subcomplex is pure of the same dimension as $\t$ (and possibly empty).  

We define an input-less task $(\s, L, \Delta)$ by setting 
$\Delta(\t) = L \cap \Chr^k \t$ for any face $\t \subseteq \s$.  Tasks constructed like this are called {\em affine}. To depict an affine task, we can simply draw the corresponding complex $L$.

By abuse of notation, we will usually write $L$ for the affine task $(\s, L, \Delta)$. We chose the name {\em affine} because if we have a task $L$ as above, the geometric realizations of the simplices of $L$ can be depicted as lying on affine subspaces of $\R^n$. Similar terminology appears in algebraic geometry, where one talks about affine varieties.

For example, consider the task of {\em total order} $L^{\ord}$,
defined as follows. For each permutation $\alpha$ of $\{0,1,\dots, n\}$, there is a unique $n$-simplex $\sigma_{\alpha}$ in the second chromatic subdivision $\Chr^2 \s$ with the property that the vertex of $\sigma_{\alpha}$ colored $i$ is in the interior of an $i$-dimensional face of $\s$. For example, for $3$ processes, the six simplices of the form $\sigma_{\alpha}$ are those shown here:
$$ \scalebox{.85}{\begin{picture}(0,0)%
\includegraphics{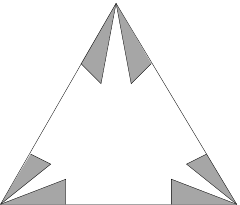}%
\end{picture}%
\setlength{\unitlength}{829sp}%
\begingroup\makeatletter\ifx\SetFigFont\undefined%
\gdef\SetFigFont#1#2#3#4#5{%
  \reset@font\fontsize{#1}{#2pt}%
  \fontfamily{#3}\fontseries{#4}\fontshape{#5}%
  \selectfont}%
\fi\endgroup%
\begin{picture}(5424,4619)(754,-5343)
\end{picture}%
}$$

The total order task is the affine task associated to the complex
$L^{\ord} \subseteq \Chr^2 \s$ is the union of all the $(n+1)!$
simplices of the form $\sigma_{\alpha}$. The name total order refers
to the fact that  the possible outputs (when all $n+1$ processes are
running) are in one-to-one correspondence with the total orderings
(i.e., permutations) of the set of processes $\{0, 1, \dots, n\}$,
similar to \emph{one-shot total-order broadcast}~\cite{HT94} 
where every process proposes
broadcasts its identifier and the processes agree on the order in
which the identifiers are delivered. 

\subsection{Views with input}
Let $r = S_1, S_2, \dots$ be a run in IIS. Recall that in
Section~\ref{sec:IIS} we defined the $k$th view of a process in the
run, $\view(p_i, k)$.  We now generalize this to allow arbitrary inputs. 

Let $\I$ be a pure $n$-dimensional chromatic complex, and let $\omega \in \I$ be an $n$-dimensional simplex. The $k$th view of $p_i$ in the run $r$ starting with input $\omega$ is denoted $\view(p_i,\omega, k)$ and defined recursively as follows:
\begin{enumerate}
\item $\view({p_i}, \omega, 0)=\{(p_i, v)\}$, where $v$ is the vertex colored $i$ in the simplex $\omega$;
\item For $k \geq 1$, the view of ${p_i} \in S_k^j \subseteq S_k$ is $\view({p_i}, \omega, k)=\{\view({p_s}, \omega, {k-1}) \mid p_s \in S_k^1 \cup \dots \cup S_k^j \}$.
\end{enumerate}

\subsection{Task Solvability}
In a sub-IIS model, informally, a task $T = (\I, \O, \Delta)$ is solvable in $M$ if for all runs $r \in M$, the infinitely participating processes output, and their output is a subsimplex of the allowed outputs for the participating processes. An output is the result of a {\em protocol}. For us, when dealing with solvability rather than complexity, a protocol is just a partial map from views to outputs. Thus, requiring an infinitely participating process to output means requiring that eventually it will have a view that is mapped by the protocol to an output value.

We define the set $\V= \V(\I)$ to consist of all possible $\view(p_i, \omega, k)$ in all runs $r \in \Runs$, for all processes $p_i$, simplices $\omega \in \I$, and integers $k \geq 0$. Formally, a protocol $\Pi$ for the task $T$ is a map from a subset of $\V$ to the set of vertices in the output complex $\O$.

\begin{definition}
A task $T = (\I, \O, \Delta)$ is {\em solvable} in a sub-IIS model $M$ if there exists a protocol $\Pi$ for $T$ such that for all $r \in M$ (with $r=S_1, S_2, \dots$ as before):
\begin{enumerate}
\item For each $p_i$, and for each $n$-dimensional simplex $\omega \in \I$, there exist $k_0$ and a vertex $v$ of $\O$ colored $i$, such that:
\begin{itemize}
\item For all $k <k_0$,  $\view(p_i, \omega, k) \notin \textit{domain}(\Pi)$;
\item For all $k \geq k_0$ such that ${p_i} \in S_k$ exists, we have $\Pi(\view({p_i}, \omega, k))=v$.
\end{itemize}

\noindent (This condition is satisfied vacuously if $p_i$ is not infinitely participating, because we can find $k_0$ such that $p_i$ did not take $k_0$ steps in $r$, so $p_i \not \in S_k$ for $k \geq k_0$.)
 \vskip5pt
 
\item For all $k$,  $\{\Pi(\view(p_i, \omega, k)) \mid \view(p_i,\omega, k)\in \textit{domain}(\Pi)$\} is a sub-simplex of a simplex in $\Delta\bigl (\omega \cap \chi^{-1}(\part(r)) \bigr)$.
\end{enumerate}
\end{definition}

In every run $r\in M$, condition (1) above requires every infinitely
participating to eventually produce an output, and condition (2)
requires the produced output to respect the task specification
$\Delta$ given the inputs of participating processes.

\subsection{Example: solving tasks in sub-IIS}
Note that our definition of task solvability in sub-IIS models brings
illuminating subtleties that were not observed in the conventional SM model.   
Consider a sub-IIS model $M$ and the corresponding model
$M_{\fast}=\{r | \exists r'\in M, r=\minimal(r')\}$.
If a task $T$ is solvable in $M$ then it is obviously solvable in
$M_{\fast}$, but not necessarily vice-versa. Indeed, consider the
obstruction-free model $\OF=\OF_1$, consisting of runs with a single
fast process. Obviously, the total order task $L^{\ord}$ cannot be
solved in $\OF$, because in runs $r$ where the process
in $\fast(r)$ is always ahead of the rest ($S_k=\fast(r)$ for all $k$), the rest of the
processes essentially proceed \emph{wait-free}. 
In contrast, we can easily solve $L^{\ord}$ in $\OF_{\fast}$ using
commit-adopt~\cite{Gaf98} (implemented in IIS).

\section{Topological interpretation}
\label{sec:top}

Recall that $\Runs$ denotes the set of runs in IIS. We put a metric on $\Runs$ as follows. Given runs $r, r' \in \Runs$, we let $k=k(r, r')$ denote the largest $k \geq 0$ such that the first $k$
steps of $r$ and $r'$ are identical. (In particular, we let $k=\infty$
when $r=r'$.) 
We set the distance between $r$ and $r'$ to be $d(r, r')
= 1/(1+k)$. 
It is easy to see that $d(r,r')$ is a metric that captures how
``close'' the two runs are. 

Recall that a metric space is compact if every open cover has a finite subcover; or, equivalently, if any infinite sequence has a convergent subsequence. (See \cite[Section 7.7]{Royden}, for example.) 
For future reference, we mention:

\begin{lemma}
\label{lem:compact}
The metric space $\Runs$ is compact.
\end{lemma}

\begin{proof} 
Let $r[1] = (S[1]_1, S[1]_2, \dots$),  $r[2]=(S[2]_1,$ $S[2]_2,$ $\dots)$, $\dots$
be an infinite sequence of runs.  (Each $S[i]_k$ is a set equipped with an ordered partition.) 
There are only finitely many possibilities for the first step $S^j_1$,
so we can find a subsequence $(r[1,1], r[1,2], \dots)$ of $(r[1],
r[2], \dots)$ such that the first step in each $r[1, i]_q$ is a
constant choice $S_1$, with a constant partition $S_1 = S_1^1 \cup
\dots \cup S_1^{n_1}$. From the subsequence $(r[1,1], r[1,2],$ $\dots)$
we can extract a further subsequence $(r[2,1], r[2,2], \dots)$ such
that the second step is a constant choice $S_2 = S_2^1 \cup \dots \cup
S_2^{n_2}$, and so on. Hence, the diagonal subsequence $(r[1,1],
r[2,2]$, $r[3,3], \dots)$ converges to $r=(S_1, S_2,
\dots)$. Thus, every sequence in $\Runs$ has a converging
subsequence. 
\end{proof}

The metric space $\Runs$ is not easy to
visualize. We can however get a partial understanding by focusing on the views of the fast processes in each run. 

Consider the standard chromatic subdivisions of the $n$-simplex
$\s$. Recall that a run in IIS can be identified with a sequence of simplices $\sigma_0, \sigma_1, \sigma_2, \dots$, with $\sigma_k \in \Chr^k \s$ and $|\sigma_{k+1}| \subset |\sigma_k|$~\cite{Lin10,Koz12}. 

Note that every run converges to a point of the geometric realization
$|\s|$, so there is a natural, continuous map $\pi: \Runs \to |\s|$,
which we call the {\em affine projection}. The information captured in
$p=\pi(r)$ exactly consists of the views of the fast processes in
$r$. In fact, each point $\pi(r) \in |\s|$ can be identified with the
minimal run $\minimal(r)$. 

There is a canonical coloring map $\chi: |\s| \to 2^{\{0,1,\dots,
  n\}}$ that extends the colorings on all chromatic subdivisions
$\Chr^m \s$ to $|\s|$. Precisely, given a point $p \in |\s|$, we let
$\chi(p)$ be the minimal subset $A \subseteq \{0,1, \dots, n\}$ such
that $p$ lies in a simplex $\sigma$ of a chromatic subdivision $\Chr^k
\s$ with $\chi(\sigma)=A$. It is easy to see that $\chi(p)=\fast(r)$, 
for any $r$ such that $\pi(r)=p$. 

A special case of a sub-IIS model is a set of runs of the form  $\pi^{-1}(S)$, where
$S \subseteq |\s|$.  We call such models 
{\em geometric}, because they can be easily visualized
as associated subsets of $|\s|$. Notice that all models in
Examples~\ref{ex:wf}-\ref{ex:adv} are geometric.  
However, our main results
will apply equally well to {\em all} (not necessarily geometric) sub-IIS models.

\section{The Generalized Asynchronous Computability Theorem}
\label{sec:gact}

\subsection{Terminating subdivisions}
Let $C$ be a chromatic complex. Consider 
standard chromatic subdivisions $\Chr^m C$ for $m >0$ (Section~\ref{sec:tasks}), and 
recall that the vertices of $\Chr^m C$ can be identified with a subset of the vertices in $\Chr^{m+1} C$.

A {\em terminating subdivision} 
$\T$ of $C$ is specified by a sequence of chromatic complexes $C_0,
C_1, C_2, \dots,$ and a sequence of subcomplexes $\Sigma_0 \subseteq \Sigma_1 \subseteq 
\Sigma_2 \subseteq \dots$ such that for all $k\geq 0$:
\begin{enumerate}[(i)]
\item $\Sigma_k$ is a subcomplex of $C_k$;
\item $C_0 = C$, and $C_{k+1}$ is obtained from $C_k$ by taking the partial chromatic
  subdivision in which the simplices in $\Sigma_k$ are ``terminated'', i.e., not further 
  subdivided. Precisely, we replace a simplex $\sigma$ in $C_k$ by a coarser subdivision than $\Chr(\sigma)$. Whereas the vertices of $\Chr(\sigma)$ are pairs $(p, \tau)$ with $\tau$ being a face of $\sigma$ and $p$ a vertex of $\tau$, in $C_{k+1}$ we consider the pairs $(p, \tau)$ of that form such that either $\tau \not \in \Sigma_k$, or $\tau$ consists of a single vertex in $\Sigma_k$. 
 For example, if $\Sigma_k$ is zero-dimensional, then $C_{k+1} = \Chr^1
C_k$; if $C_k$ is the standard $2$-dimensional simplex and $\Sigma_k$
is one of its $1$-dimensional faces, we have:
$$ \scalebox{.85}{\begin{picture}(0,0)%
\includegraphics{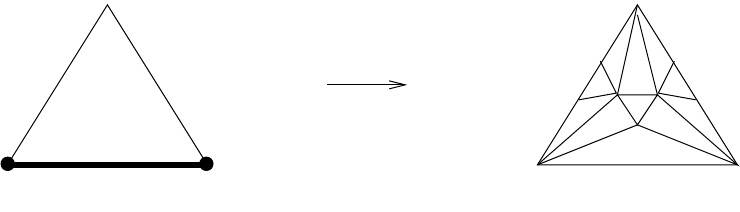}%
\end{picture}%
\setlength{\unitlength}{2526sp}%
\begingroup\makeatletter\ifx\SetFigFont\undefined%
\gdef\SetFigFont#1#2#3#4#5{%
  \reset@font\fontsize{#1}{#2pt}%
  \fontfamily{#3}\fontseries{#4}\fontshape{#5}%
  \selectfont}%
\fi\endgroup%
\begin{picture}(5543,1563)(2045,-3712)
\put(2101,-2461){\makebox(0,0)[lb]{\smash{{\SetFigFont{10}{12.0}{\rmdefault}{\mddefault}{\updefault}{\color[rgb]{0,0,0}$C_k$}%
}}}}
\put(5851,-2461){\makebox(0,0)[lb]{\smash{{\SetFigFont{10}{12.0}{\rmdefault}{\mddefault}{\updefault}{\color[rgb]{0,0,0}$C_{k+1}$}%
}}}}
\put(2811,-3621){\makebox(0,0)[lb]{\smash{{\SetFigFont{10}{12.0}{\rmdefault}{\mddefault}{\updefault}{\color[rgb]{0,0,0}$\Sigma_k$}%
}}}}
\end{picture}%
}$$
\end{enumerate}

A simplex of $\Sigma_k$ for some $k$ is called a {\em stable simplex}
in the subdivision $\T$; such a simplex remains unchanged in all
further complexes $C_{k+1}, C_{k+2}, \dots$. The union $\cup_k \Sigma_k$ of stable
simplices in $\T$ forms a chromatic complex, which
we denote by $K(\T)$; it usually has infinitely many vertices. 
Observe that the geometric realization $|K(\T)|$ can be identified with a subset of $|C|$.

In particular, if there exists $k$ such that $\Sigma_k=C_k$, then we must have $C_k = C_{k+1} = \dots$, and $\T$ is just a finite subdivision of $C$; in this case, all the simplices in $C_k$ are stable, and $|K(\T)| = |C|$. At the other extreme, if $\Sigma_k$ is empty for all $k$, then $\T$ consists of the $k\th$ chromatic subdivisions of $C$ for all $k$; in this case, no simplices are stable, and $K(\T)$ is empty.

Stable simplices intend to model processes that have produced outputs 
and thus, intuitively, do not need to communicate among themselves any
longer. Therefore, stable simplices are not further
subdivided. However,  processes with outputs keep participating in
the computation: simplices that contain non-stabilized vertices
continue to be subdivided.  This will allow us to formulate the
conditions of task solvability in non-compact sub-IIS models.    
 

\subsection{The main result}

We are now ready to formulate and prove our main result: a
characterization of taks solvability in sub-IIS models.

Let $M \subseteq \Runs$ be a sub-IIS model on $n+1$
processes. 
Let $\T$ be a terminating subdivision of a pure $n$-dimensional chromatic complex $\I$, and let $\chi: \I \to \s$ be the coloring map. Let $\rho : |\I| \to |\s|$ be the geometric realization (piecewise linear extension) of the map $\chi$. Note that $\rho$ maps vertices of $\Chr^k \I$ to vertices of $\Chr^k \s$ of the same color. 
 
Recall that each vertex $v$ in $\T$ belongs to $\Chr^k(\sigma)$ for some
$k\geq 0$ and some $n$-dimensional simplex $\sigma$ of $\I$. Thus, $\rho(v)$ is a vertex of $\Chr^k(\s)$. If we have a simplex $\tau \in K(\T)$, then $\rho(|\tau|)$ is the convex hull of $\rho(v)$ for $v \in \tau$.

We say that $\T$ is {\em admissible} for $M$ if for 
any run $r \in M$ (viewed as a sequence of simplices $\sigma_0,
\sigma_1, \dots$ in $\s$) and for every $n$-dimensional simplex $\omega$ in $\I$, there exists $k > 0$ and a stable
simplex $\tau\in K(\T)$ such that $|\tau| \subseteq |\omega|$ and $|\sigma_k| \subseteq\rho(|\tau|)$.
The intuition here is that every run of $M$ with inputs $\omega$
should eventually land in a simplex of $\T$.  

\begin{theorem}[GACT]
\label{thm:gact}
A sub-IIS model $M$ solves a task $T=(\I, \O, \Delta)$ if and only if there exists a terminating subdivision $\T$ of $\I$ and a chromatic map $\delta: K(\T) \to \O$ such that:
\begin{enumerate}[(a)]
\item 
$\T$ is admissible for the model $M$;
\item For any simplex $\sigma$ of $\I$, if $\tau$ is a stable simplex of $\T$ such that $|\tau| \subseteq |\sigma|$, then $\delta(\tau) \in \Delta(\sigma)$.
\end{enumerate}

\end{theorem}

\begin{proof}
''$\Rightarrow$'' : Suppose $M$ solves $T$ using a protocol $\Pi$. By
induction on the recursion that defines $\view({p_i}, \omega, k)$, it
is easy to see that the $k\th$ view of $p_i$ in a run $r\in M$ (with input
$\omega$) corresponds to a vertex  $v \in \Chr^k(\omega) \subseteq \Chr^k(\I)$ with $\chi (v)=\{p_i\}$. 

We construct a terminating subdivision $\T$ with desirable properties
as follows. We proceed with the standard subdivisions
$\Chr^k(\I)$ for $k=0,1,2, \ldots ,$ and we examine all runs $r \in
M$. At the $k\th$ stage we take the inductively constructed $C_k$,
whose vertices are a subset of the vertices of $\Chr^k(\I)$. We then
terminate those simplices for which $\Pi$ has given an output: A
simplex $\sigma$ of $C_k$ (with $|\sigma| \subseteq |\omega|$) is included in $\Sigma_k$ if there exists a
run $r \in M$ such that the vertices of $\sigma$ are of the form
$v_i=\view({p_i}, \omega, k)$ for that run, and the outputs $\Pi(v_i)$ exist (that is, $v_i\in \mathit{domain}(\Pi)$).  
Then $\Sigma_k$ determines $C_{k+1}$.

Given a simplex $\sigma \in \Sigma_k$ with vertices $v_i$, we set $\delta(\sigma)$ be the simplex with vertices $\Pi(v_i)$. 

Part (a) (admissibility  of $\T$) follows from the construction: Given
any run $r \in M$ and a top-dimensional simplex $\omega$ in $\I$, pick
$k$ such that all the processes infinitely participating in $r$ have
produced output at the $k\th$ step when they are given input from
$\omega$. Let $\sigma_k$ be the corresponding simplex of $\Chr^k \s$.
If $\rho^{-1}(|\sigma_k|) \cap |\omega|$ is an embedded simplex of $C_k$, 
then it is necessarily a stable simplex (because all the
processes have output), and we are done. 
If $\rho^{-1}(|\sigma_k|) \cap |\omega|$ is not an embedded simplex of
$C_k$, then, by construction, it is contained in a simplex of $C_k$  that was
terminated before (because some of the processes have produced outputs at an earlier time).

Part (b) follows from the fact that $M$ solves $T$ using $\Pi$.
\medskip

``$\Leftarrow$'' : Conversely, suppose there exists a terminating
subdivision $\T$ and a map $\delta: K(\T) \to \O$ as in the statement
of the theorem. We construct a protocol $\Pi$ 
by which $M$ solves $T$.
Suppose we have a run $r \in M$, corresponding to a sequence of simplices $\sigma_0 \subseteq \sigma_1 \subseteq \sigma_2 \subset \ldots$ 
Since $\T$ is admissible for $M$, for each input $\omega$ there exists
a stable simplex $\tau$ such that  $|\tau| \subseteq |\omega|$ and
$|\sigma_k| \subseteq \rho (|\tau|)$ for all $k \gg 0$.  Given a
process $p_i \in \ipart(r)$, we can assign it as output value the
vertex of $\delta(\tau)$ that has color $p_i$. Now, $p_i$ may obtain
an output value (necessarily the same as before) through another run,
at a different step $k$. We take the minimum over all such $k$, to
obtain the value $k_0$ needed in the definition of $\Pi$. 
Condition (b) implies that $\Pi$ solves $T$. 
\end{proof}

\section{The wait-free model}
\label{sec:wf}

For the wait-free model $WF$, let us see how we can derive the original Asynchronous
Computability Theorem of~\cite{HS99}. Indeed, Theorem~\ref{thm:gact} has the following:
\begin{corollary}
\label{cor:act}
A task $T=(\I, \O, \Delta)$ is solvable in the wait-free model if and only if there exists $k \geq 0$ and a chromatic map $\eta: \Chr^k \I \to \O$ such that, for any simplex $\tau \subseteq \I$ and any subsimplex $\sigma$ of $\Chr^k \tau \subset \Chr^k \I$, we have $\eta(\sigma) \in \Delta(\tau)$.
\end{corollary}

\begin{proof}
If $\eta: \Chr^k \I \to C$ exists, solvability of $T$ follows from GACT because $\Chr^k \I$ (with all the vertices terminated at the $k\th$ step) is a terminating subdivision that is admissible for $\WF$. 

Conversely, suppose that $T$ is wait-free solvable. GACT provides a terminating subdivision $\T$ that is admissible for $\WF$, and a map $\delta: K(\T) \to \O$. For each run $r \in \WF=\Runs$ and top-dimensional input simplex $\omega \in I$, there is a $k$ such that $|\sigma_k|$ is contained in a stable simplex $\tau$ of $\T$ with $|\tau| \subseteq |\omega|$. Since there are finitely many possibilities for $\omega$, we can find a $k=k(r)$ that works for all $\omega$.  Let $\Runs_k$ be the set of runs $r$ for which $k(r) \leq k$. We have inclusions $\Runs_0 \subseteq \Runs_1 \subseteq \Runs_2 \subseteq \dots$, and each $\Runs_k$ is open in the metric on $\Runs$ introduced in Section~\ref{sec:top}. We know from Lemma~\ref{lem:compact} that the set $\Runs$ is compact. Hence, the open cover ${\Runs_k}$ of $\Runs$ admits a finite subcover; i.e., there exists $k$ such that $\Runs_k = \Runs$. We now define the desired map $\eta: \Chr^k \I \to \O$ by setting $\eta(\sigma) =\delta(\tau)$, where $\tau \in K(\T)$ is the minimal simplex with $|\sigma| \subseteq |\tau|$.
\end{proof}

As stated in \cite{HS99}, ACT characterizes solvability in terms of a map from an arbitrary colored subdivision of $\I$ to the output complex. That any colored subdivided simplex can be approximated by $\Chr^k(\I)$ for some $k$ large enough is a purely topological result, proved in \cite{HS99}, and which can be used here verbatim. (This corresponds to Step 4 in the outline of the proof of ACT from the Introduction.)

\section{Simplicial approximation}
\label{sec:proofs}

To be able to apply GACT, we need a tool for constructing chromatic
maps between two chromatic complexes $A$ and $B$ (subject to some
boundary conditions). In many cases, it is easier
to first construct a continuous map $f: |A| \to |B|$.  Standard
results in algebraic topology (reviewed in Subsection~\ref{sec:classical} below) say that after replacing $A$ by a fine enough subdivision, we can deform $f$ into a geometric realization of
a simplicial map. Such a map may collapse the dimension of simplices, so it is not always clear how to turn it into a chromatic map. However, if we impose an additional condition (link-connectedness for the target), we will show in Subsection~\ref{sec:ca} that one can do the approximation using chromatic maps.

\subsection{Classical results}
\label{sec:classical} 
Let $A$ and $B$ be simplicial complexes and let $f: |A| \to |B|$ be a continuous map between their geometric realizations. If $A'$ is a subdivision of $A$, a simplicial map $\phi: A' \to B$ is called a {\em simplicial approximation} to $f$ if for every $x \in |A| = |A'|$ and $\sigma \in B$ we have
 $$f(x) \in |\sigma| \ \Rightarrow \ |\phi|(x) \in |\sigma|.$$

Roughly, the simplicial approximation theorem says that every continuous map between simplicial complexes can be approximated by a simplicial map. There are several versions of this in the literature. For finite simplicial complexes, we have:

\begin{theorem}
\label{thm:fSA}
Let $A$ and $B$ be simplicial complexes such that $A$ is finite, and let  $f: |A| \to |B|$ be a continuous map. Then: 
\begin{enumerate}[(a)]
\item There exists an integer $N$ such that for all $n \geq N$, the map $f$ admits a simplicial approximation $\phi: \Bary^n(A) \to B$. 

\item Furthermore, if we have a subcomplex $C \subseteq A$ such that the restriction of $f$ to $C$ is the geometric realization of a simplicial map $g: C \to B$, then the approximation $\phi$ can be taken so that the restriction of $\phi$ to $|C|$ equals $|g|$. 
\end{enumerate}
\end{theorem}

Part (a) of this result is a special case of Theorem 8 in \cite[p.128]{Spanier}. The theorem is stated in \cite{Spanier} in more generality, for pairs of simplicial complexes. Part (b) above follows from this more 
statement, taking into account Lemma 1 in \cite[p.126]{Spanier}.

In this paper we will need a different variant of the simplicial approximation theorem, one that applies without the hypothesis that $A$ is finite:
\begin{theorem}
\label{thm:infSA}
Let $A$ and $B$ be simplicial complexes, and let  $f: |A| \to |B|$ be a continuous map. Then: 
\begin{enumerate}[(a)]
\item There exists a subdivision $A'$ of $A$ such that the map $f$ admits a simplicial approximation $\phi: A' \to B$. 
\item Furthermore, if we have a subcomplex $C \subseteq A$ such that the restriction of $f$ to $C$ is the geometric realization of a simplicial map $g: C \to B$, then the approximation $\phi$ can be taken so that the restriction of $\phi$ to $|C|$ equals $|g|$. 
\end{enumerate}
\end{theorem}

Theorem~\ref{thm:infSA} is mentioned in the remarks at the bottom of p.128 in \cite{Spanier}; see \cite{Whitehead} or \cite{Munkres84} for a more complete treatment. 

Note that, in the case when $A$ is countable and locally finite, we can deduce Theorem~\ref{thm:infSA} from Theorem~\ref{thm:fSA} as follows. Let us write $A$ as a union of finite simplicial complexes $A_1 \subseteq A_2 \subseteq \dots$ We construct the subdivision $A'$ and the map $\phi$ inductively. Suppose we found a simplicial approximation $\phi_k : \Bary^{n_k}(A_k) \to B$ for the restriction of $f$ to $A_k$. Consider the restriction of $f$ to $|A_{k+1}| = |\Bary^{n_k}(A_{k+1})|$. We extend the approximation $\phi_k$ to $A_{k+1}$ by using part (b) of Theorem~\ref{thm:fSA}, applied to $\Bary^{n_k}(A_{k+1})$ and $B$. The result is a simplicial approximation $\phi_{k+1}: \Bary^{n_{k+1}}(A_{k+1}) \to B$ for some $n_{k+1} \geq n_k$, such that $|\phi_k|=|\phi_{k+1}|$ on $|A_k|$. The desired approximation $\phi: A' \to B$ has $|\phi| = |\phi_k|$ on each $|A_k|$. A subtle point here is the construction of the subdivision $A'$, which is getting finer and finer as we go towards infinity. In principle, we would like $A'$ to be $\Bary^{n_k}(A_k)$ on each $|A_k| \setminus |A_{k-1}|$. This is not a simplicial complex, but we can turn it into one by introducing additional simplices, as shown in the figure:
$$\scalebox{.8}{\begin{picture}(0,0)%
\includegraphics{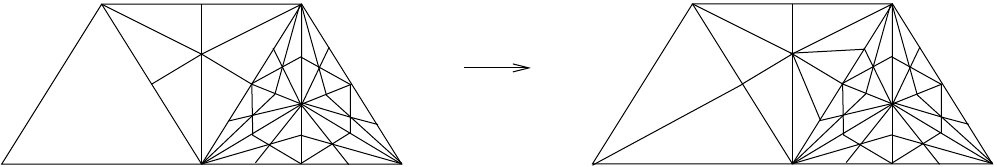}%
\end{picture}%
\setlength{\unitlength}{2526sp}%
\begingroup\makeatletter\ifx\SetFigFont\undefined%
\gdef\SetFigFont#1#2#3#4#5{%
  \reset@font\fontsize{#1}{#2pt}%
  \fontfamily{#3}\fontseries{#4}\fontshape{#5}%
  \selectfont}%
\fi\endgroup%
\begin{picture}(7458,1234)(4564,-3381)
\end{picture}%
}$$
The local finiteness of $A$ ensures that there is an upper bound on the number of times we have to subdivide each simplex.

\subsection{Chromatic approximations}
\label{sec:ca}

Let us go back to Theorem~\ref{thm:fSA}. Note that if $A$ is a chromatic complex, then instead of the barycentric subdivisions $\Bary^n(A)$, one could take standard chromatic subdivisions $\Chr^n(A)$. (The same proof applies.) However, we cannot a priori conclude that the simplicial approximation is a chromatic simplicial map. For example, if the continuous map $f$ collapses a simplex of $A$ to a single vertex in $B$, then any simplicial approximation would do the same, but on the other hand chromatic maps are non-collapsing.

Nevertheless, we can avoid collapsing by assuming that the following property (for the target complex $B$):

\begin{definition}[Definition 4.14 in \cite{HS99}] \label{def:link}
A pure $n$-dimensional complex $B$ is called {\em link-connected} if for all simplices $\sigma \in B$, the link of $\sigma$ in $B$ is $(n-\dim(\sigma)-2)$-connected. 
\end{definition}

For example, the output complex $L^{\ord}$ for the total order task on three processes is not link-connected, because the link (in $L^{\ord}$) of a vertex of $\s$ is not connected. 

 A variant of Theorem~\ref{thm:infSA} for chromatic maps is proved in \cite[Lemma 4.21]{HS99} under the assumptions that $A$ and $B$ are chromatic complexes, $B$ is link-connected, $A$ is a finite subdivision of the standard simplex, and $C$ is the boundary of $A$. The conclusion is that the map $g$ can be taken to be chromatic. Furthermore, Theorem 5.29 in \cite{HS99}
shows that, under the same assumptions, the subdivision $A'$ of $A$ can be taken to be a standard chromatic subdivision; this yields a chromatic variant of Theorem~\ref{thm:fSA}.

One can generalize the results of Herlihy and Shavit to the setting of infinite complexes:

\begin{theorem}
\label{thm:chrSA}
Let $A$ and $B$ be chromatic simplicial complexes, and let  $f: |A| \to |B|$ be a continuous map. Suppose that $A$ is countable and locally finite, and that $B$ is link-connected. Then: 
\begin{enumerate}[(a)]
\item There exists a subdivision $A'$ of $A$ such that the map $f$ admits a chromatic simplicial approximation $\phi: A' \to B$. 
\item Furthermore, if we have a subcomplex $C \subseteq A$ such that the restriction of $f$ to $C$ is the geometric realization of a chromatic simplicial map $g: C \to B$, then the approximation $\phi$ can be taken so that the restriction of $\phi$ to $|C|$ equals $|g|$. 
\end{enumerate}
\end{theorem}

\begin{proof}
We do this inductively on the skeleta of $A$. Suppose we have defined the map $\phi$ on a subdivision the $k$-skeleton $\Skel^k(A)$. We apply \cite[Lemma 4.21]{HS99} to the restriction of $f$ to each $(k+1)$-simplex $\sigma$ of $A$, mapped to the $(k+1)$-skeleton of $B$. (Observe that if $B$ is link-connected, then so are its skeleta.) We obtain a simplicial approximation to $f$ on $\sigma$, agreeing with the already constructed approximation on the boundary of $\sigma$. In the process we have to subdivide the simplices $\sigma'$ in $\Skel^k(A)$, and each $\sigma'$ is on the boundary of several $(k+1)$-simplices $\sigma$. However, by local finiteness, we can find a sufficiently fine subdivision that works for all $\sigma \supset \sigma'$. By continuing this ad infinitum, we obtain the approximation $\phi$. Furthermore, if we have a subcomplex $C$ as in part (b), then at each step we arrange so that the approximation agrees with the one defined on the corresponding skeleton of $C$. 
\end{proof}

\section{Link-connected tasks}
\label{sec:lc}

\subsection{A general result} 
The following proposition is an extension of the work of Herlihy and
Shavit from~\cite{HS99} to the case of \emph{infinite} chromatic complexes. 

\begin{proposition}
\label{prop:ext}
Let $M$ be a sub-IIS model, and $\T$ a terminating subdivision of $\I$ that is admissible for $M$. Suppose we have a task $T=(\I, \O, \Delta)$ such that the complexes $\Delta(\tau)$ are link-connected, for any $\tau \in \I$. Then, the task $T$ is solvable in $M$ if and only if there exists a continuous map $f: |K(\T)| \to |\O|$ such that $f(|K(\T)| \cap |\tau|) \subseteq |\Delta(\tau)|$ for all $\tau \in \I$.
\end{proposition}

\begin{proof} If $\T$ and $\T'$ are terminating subdivisions of $\s$, we say that $\T'$ is a {\em stable refinement} of $\T$ if $|K(\T)| = |K(\T')|$, and every simplex of $\T'$ is contained in a simplex of $\T$; i.e., $K(\T')$ should be a subdivision of $K(\T)$. Note that if $\T$ is admissible for a model $M$, then so is $\T'$.

Given the continuous map $f$, we shall construct a simplicial, chromatic approximation $\delta: K(\T') \to \O$ as needed to apply GACT; here, $\T'$ is a stable refinement of $\T$. 

We first construct a chromatic subdivision $K'$ of $K(\T)$, whose vertices are not necessarily in the standard chromatic subdivisions of $\I$, and a chromatic map $\delta': K' \to \O$ (an approximation to $f$) such that $\delta'$ is carrier-preserving: $\delta'(\sigma) \in \Delta(\tau)$ when $|\sigma| \subseteq |\tau|$. We do this inductively on $d \geq 0$: For each $d$, we define the values of $\delta'$ on the simplices that are contained in $d$-dimensional faces of $|\I|$. Suppose we have defined $\delta'$ for $d-1$, and pick a $d$-dimensional face $\tau$ of $|\I|$. The restriction of $f$ to $|K(\T)| \cap |\tau|$ can be approximated by a simplicial map from a subdivision of $K(\T)$, extending the already constructed $\delta'$ on the $(d-1)$-dimensional boundary. Further, since $K(\T)$ is locally finite (by definition) and $\Delta(\tau)$ is link-connected, it follows from Theorem~\ref{thm:chrSA} that we can arrange for $\delta'$ to preserve colors. 

Thus, we find a sufficiently fine stable refinement $\T'$ of $\T$ and a chromatic, carrier-preserving map $g: K(\T') \to K'$. We then set $\delta = \delta' \circ g$ and apply Theorem~\ref{thm:gact}.

Conversely, if $T$ is solvable in $M$, we can apply GACT and obtain a terminating subdivision $\T$ and a chromatic map $\delta: K(\T) \to \O$. The desired continuous map $f$ is the geometric realization of $\delta$.
\end{proof}

\subsection{An example of GACT in action}
\label{sec:action}
Consider the $t$-resilient model $\Res_t$ from Example~\ref{ex:res}. Let $L_t$ be the affine task with output complex consisting of all the simplices $\sigma$ in the second chromatic subdivision $\Chr^2 \s$ such that no vertex of $\sigma$ is on an $(n-t-1)$-dimensional face of $\s$. For example, when $n=2$ and $t=1$, the output complex for $L_1$ looks like:
$$ \scalebox{.85}{\begin{picture}(0,0)%
\includegraphics{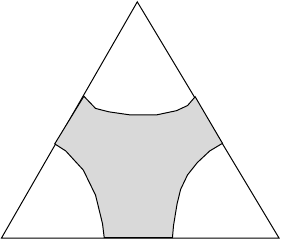}%
\end{picture}%
\setlength{\unitlength}{2486sp}%
\begingroup\makeatletter\ifx\SetFigFont\undefined%
\gdef\SetFigFont#1#2#3#4#5{%
  \reset@font\fontsize{#1}{#2pt}%
  \fontfamily{#3}\fontseries{#4}\fontshape{#5}%
  \selectfont}%
\fi\endgroup%
\begin{picture}(2139,1824)(754,-2548)
\put(1731,-2034){\makebox(0,0)[lb]{\smash{{\SetFigFont{7}{8.4}{\rmdefault}{\mddefault}{\updefault}{\color[rgb]{0,0,0}$L_1$}%
}}}}
\end{picture}%
}$$

\begin{proposition}
\label{prop:res}
The task $L_t$ is solvable in the model $\Res_t$.
\end{proposition}

\begin{proof} Note that for each face $\t \subseteq \s=\I$, the complex $\Delta(\t)$ for the task $L_t$ is link-connected. Therefore, it suffices to find a terminating subdivision $\T$ and a continuous map $f$ with the properties required in Proposition~\ref{prop:ext}.

For $n \geq 0$, let $\tilde R_n \subset |\s|$ be the union of (the geometric realizations of) all the simplices $\sigma \subset \Chr^{n+2} \s$ such that no vertex of $\sigma$ is on an $(n-t-1)$-dimensional face of $\s$. Let $R_0 = |L_t|$ and, for $n > 0$, let $R_n$ be the closure of $\tilde R_n - \tilde R_{n-1}$. The union of all $R_n$'s is the complement of the $(n-t-1)$-skeleton of $\s$:
$$ \scalebox{.85}{\begin{picture}(0,0)%
\includegraphics{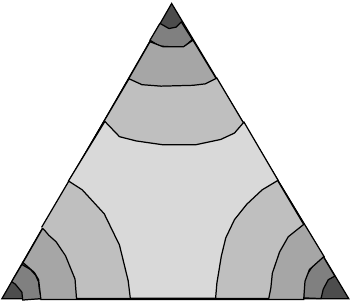}%
\end{picture}%
\setlength{\unitlength}{3108sp}%
\begingroup\makeatletter\ifx\SetFigFont\undefined%
\gdef\SetFigFont#1#2#3#4#5{%
  \reset@font\fontsize{#1}{#2pt}%
  \fontfamily{#3}\fontseries{#4}\fontshape{#5}%
  \selectfont}%
\fi\endgroup%
\begin{picture}(2139,1832)(754,-2556)
\put(1731,-1154){\makebox(0,0)[lb]{\smash{{\SetFigFont{6}{7.2}{\rmdefault}{\mddefault}{\updefault}{\color[rgb]{0,0,0}$R_2$}%
}}}}
\put(1696,-2003){\makebox(0,0)[lb]{\smash{{\SetFigFont{9}{10.8}{\rmdefault}{\mddefault}{\updefault}{\color[rgb]{0,0,0}$R_0$}%
}}}}
\put(1212,-2247){\makebox(0,0)[lb]{\smash{{\SetFigFont{8}{9.6}{\rmdefault}{\mddefault}{\updefault}{\color[rgb]{0,0,0}$R_1$}%
}}}}
\put(990,-2389){\makebox(0,0)[lb]{\smash{{\SetFigFont{6}{7.2}{\rmdefault}{\mddefault}{\updefault}{\color[rgb]{0,0,0}$R_2$}%
}}}}
\put(1726,-1461){\makebox(0,0)[lb]{\smash{{\SetFigFont{8}{9.6}{\rmdefault}{\mddefault}{\updefault}{\color[rgb]{0,0,0}$R_1$}%
}}}}
\put(2184,-2215){\makebox(0,0)[lb]{\smash{{\SetFigFont{8}{9.6}{\rmdefault}{\mddefault}{\updefault}{\color[rgb]{0,0,0}$R_1$}%
}}}}
\put(2441,-2341){\makebox(0,0)[lb]{\smash{{\SetFigFont{6}{7.2}{\rmdefault}{\mddefault}{\updefault}{\color[rgb]{0,0,0}$R_2$}%
}}}}
\end{picture}%
}$$

The terminating subdivision $\T$ is as follows: It starts with $\Sigma_0 = \Sigma_1 =
\emptyset$, so that $C_0=\s, C_1=\Chr^1 \s, C_2=\Chr^2 \s$. We then let $\Sigma_2$ be the subcomplex of $\Chr^2 \s$ supported in the region $R_0$. This defines $C_2$ so that on the complement of $R_0$ it consists of the simplices in $\Chr^3\s$. After this, we terminate all
the simplices in $R_1$. Now on the complement of $R_0\cup R_1$ we
have the fourth chromatic subdivision. We then terminate the simplices in $R_2$, and so on.  
By construction, eventually, every simplex contained in any $R_k$ is
stable in this terminating subdivision. Since the affine projection
$\pi(\Res_t)$ is contained in the union $|K(\T)|$ of all the $R_k$'s, we deduce that $\T$ is admissible for $\Res_t$.

It remains to construct the continuous map $f: |K(\T)| \to |L_t|=R_0$. We let the restriction of $f$ to $R_0$ be the identity, and map everything else onto the boundary $R_0 \cap R_1$ using radial projection away from the $(n-t-1)$-skeleton of $\s$:
$$ \scalebox{.85}{\begin{picture}(0,0)%
\includegraphics{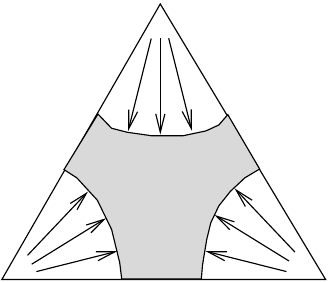}%
\end{picture}%
\setlength{\unitlength}{2901sp}%
\begingroup\makeatletter\ifx\SetFigFont\undefined%
\gdef\SetFigFont#1#2#3#4#5{%
  \reset@font\fontsize{#1}{#2pt}%
  \fontfamily{#3}\fontseries{#4}\fontshape{#5}%
  \selectfont}%
\fi\endgroup%
\begin{picture}(2139,1824)(754,-2548)
\put(1699,-2030){\makebox(0,0)[lb]{\smash{{\SetFigFont{8}{9.6}{\rmdefault}{\mddefault}{\updefault}{\color[rgb]{0,0,0}$R_0$}%
}}}}
\end{picture}%
}$$
Observe that radial projection preserves boundaries, so Proposition~\ref{prop:ext} applies. \end{proof}

\begin{remark}
An alternative, operational solution of  task $L_t$ via
$t$-resilient atomic-snapshots was given by the
first author in~\cite{G98}. More precisely, the Red-Yellow-Green
algorithm in~\cite[Section~4]{G98} specifies an intricate simulation
scheme that allows for solving $L_t$.
\end{remark}

\section{Related work}
\label{sec:related}

The topological conditions of wait-free task solvability were
expressed by Herlihy and Shavit ~\cite{HS93,HS99}  in the form of ACT.
In the restricted case of \emph{colorless} tasks that, roughly, can be defined without taking process
identifiers in mind, Herlihy and Rajsbaum~\cite{HR10-shell,HR13} derived
task solvability conditions in adversarial shared-memory
models~\cite{DFGT11}.
This paper proposes a characterization  of generic (not
necessarily colorless) tasks in any (not necessarily adversarial) sub-IIS model.

The IIS model was introduced by Borowsky and
Gafni~\cite{BG97} and shown to precisely capture  the
standard chromatic subdivision of the input
complex~\cite{Lin10,Koz12}.
Due to the elegance of its topological representation, IIS has been
widely used topological reasoning about distributed computing~\cite{HS93,BG93b,BG97,HS99,HKR14}. 
In~\cite{BG97,GR10-opodis}, IIS has been shown equivalent to
SM in terms of task solvability.
Rajsbaum et al.~\cite{RRT08} and, more recently, Raynal and Stainer~\cite{RS12} relate proper subsets of
sub-IIS and sub-SM models restricted using specific failure detectors. 
A recent paper~\cite{BGK14} extends these equivalences to arbitrary 
sub-SM and sub-IIS models,
thus justifying the choice of IIS as a
model of study.

The difficulty of dealing with certain problems in certain non-compact models, such as
consensus and $t$-resilience, has been studied before by Lubitsch and
Moran~\cite{LM95-compact}, Brit and Moran~\cite{BM96-compact},
Moses and Rajsbaum~\cite{MR02-compact}.
By deriving topological solvability conditions for any task and any sub-IIS model, this
paper brings this work to a higher level of generality. 
The continuous space $|\s|$ has appeared previously in the work of Saks and
Zaharoglou~\cite{SZ00} where it was used to derive the impossibility
wait-free set agreement. 


\section{Concluding remarks}\label{sec:conc}

We presented a version of a generalization of ACT.
Other versions may be possible through the relation between simplicial and continuous maps,
as well as through defining terminating subdivisions not necessarily with respect to $\Chr^m(\s)$, in 
analogy to the sufficiency condition of ACT~\cite{HS99}. We chose the simplest version which 
still provides us with the benefit of producing a ``topological
solution'' to the given task.

The main technical challenge we faced was to define and view the IIS model 
directly,  rather than just through the prism of the simulations from the
standard (non-iterated) model~\cite{BG97,GR10-opodis}. 
This brought forth a coherent view of IIS,  as well as exposed the
richness of the model.

The generic IIS models considered in this paper are just arbitrary
subsets of the various possible interleaving of reads and writes,
which is an extension with respect to the previous attempts to model
sub-IIS computations~\cite{RRT08,RS12}. Yet, distributed computing
refers also to the availability of one-shot  objects, e.g., consensus,
$k$-set agreement, etc. Of course, we can produce a sub-IIS model
which is equivalent to having consensus, or any other simple
object. An open question is whether our framework embedded in a large
enough dimension can capture  ``all of distributed computing,''
at least with respect to terminating computations. In particular, what will be the
sets of runs that correspond to the availability of the M\"obius
task~\cite{GRH06} or a task from the family of $0$-$1$
exclusion~\cite{Gaf08-01}? We know that, for instance, the ``symmetric''
task on $n$ processes from~\cite{GRH06} can also be formulated as a
regular task on $2n-1$ processes, hence increasing the task dimension
does help here. Our speculation is that any computability question for a
``reasonable'' one-shot problem in distributed computing is equivalent to
a question of task solvability in a sub-IIS model. 

\subsection*{Acknowledgement}
We are in debt to Robert F. Brown for helpful discussions in the very early stages of this research.


\bibliography{references}

\end{document}